\documentclass[11pt]{article}
\usepackage{amsfonts,amsmath, amssymb, amsthm, algorithm,algpseudocode, enumitem, thm-restate}

\PassOptionsToPackage{obeyspaces}{url}
\usepackage[colorlinks=true,citecolor=blue,urlcolor=blue,linkcolor=blue,bookmarksopen=true]{hyperref}
\usepackage{cleveref}

\author{Zhiyang He\footnote{Department of Mathematics, Massachusetts Institute of Technology, Cambridge, United States.}\and Jason Li\footnote{Simons Institute, UC Berkeley, Berkeley, United States.}}
\title{Breaking the $n^k$ Barrier for Minimum $k$-cut on Simple Graphs}

\newtheorem{theorem}{Theorem}[section]

\newtheorem{claim}[theorem]{Claim}
\newtheorem*{claim*}{Claim}
\newtheorem{definition}[theorem]{Definition}

\newtheorem{lemma}[theorem]{Lemma}
\newtheorem*{lemma*}{Lemma}
\newtheorem{corollary}[theorem]{Corollary}

\begin{document}

\maketitle
\thispagestyle{empty}
\begin{abstract}
In the minimum $k$-cut problem, we want to find the minimum number of edges whose deletion breaks the input graph into at least $k$ connected components. The classic algorithm of Karger and Stein~\cite{karger1996new} runs in $\tilde O(n^{2k-2})$ time,\footnote{$\tilde O(\cdot)$ denotes omission of polylogarithmic factors in $n$.} and recent, exciting developments have improved the running time to $O(n^k)$~\cite{GHLL20}. For general, weighted graphs, this is tight assuming popular hardness conjectures.

In this work, we show that perhaps surprisingly, $O(n^k)$ is not the right answer for simple, \emph{unweighted} graphs. We design an algorithm that runs in time $O(n^{(1-\epsilon)k})$ where $\epsilon>0$ is an absolute constant, breaking the natural $n^k$ barrier. This establishes a separation of the two problems in the unweighted and weighted cases.
\end{abstract}
\newpage

\setcounter{page}{1}

\section{Introduction}
In this paper, we study the (unweighted) minimum $k$-cut problem: given an undirected graph $G=(V,E)$ and an integer $k$, we want to delete the minimum number of edges to split the graph into at least $k$ connected components. Throughout the paper, let $\lambda_k$ denote this minimum number of edges. Note that the $k$-cut problem generalizes the global minimum cut problem, which is the special case $k=2$.

For fixed constant $k\ge2$, the first polynomial-time algorithm for this problem is due to Goldschmidt and Hochbaum~\cite{GH94}, who designed an algorithm running in $n^{O(k^2)}$ time. Subsequently, Karger and Stein showed that their (recursive) randomized contraction algorithm solves the problem in $\tilde O(n^{2k-2})$ time. This was later matched by a deterministic algorithm of Thorup~\cite{Thorup08} based on tree packing, which runs in $\tilde O(n^{2k})$ time.

These algorithms remained the state of the art until a few years ago, when new progress was established on the problem~\cite{gupta2018faster,GLL19,Li19}, culminating in the $\tilde O(n^k)$ time algorithm of Gupta, Harris, Lee, and Li~\cite{GHLL20} which is, surprisingly enough, just the original Karger-Stein recursive contraction algorithm with an improved analysis. The $\tilde O(n^k)$ time algorithm also works for \emph{weighted} graphs, and they show by a reduction to max-weight $k$-clique that their algorithm is asymptotically optimal, assuming the popular conjecture that max-weight $k$-clique cannot be solved faster than $\Omega(n^{k-O(1)})$ time. However, whether the algorithm is optimal for \emph{unweighted} graphs was left open; indeed, the (unweighted) $k$-clique problem can be solved in $n^{(\omega/3)k+O(1)}$ time through fast matrix multiplication.\footnote{As standard, we define $\omega$ as the smallest constant such that two $n\times n$ matrices can be multiplied in $O(n^{\omega+o(1)})$ time. The best bound known is $\omega<2.373$~\cite{alman2021refined}, although $\omega=2$ is widely believed.)} Hence, the time complexity of \emph{unweighted} minimum $k$-cut was left open, and it was unclear whether the right answer was $n^k$, or $n^{(\omega/3)k}$, or somewhere in between.

In this paper, we make partial progress on this last question by showing that for \emph{simple} graphs, the right answer is asymptotically bounded away from $n^k$:
\begin{theorem}
There is an absolute constant $\epsilon>0$ such that the minimum $k$-cut problem can be solved in $n^{(1-\epsilon)k+O(1)}$ time.
\end{theorem}
In fact, we give evidence that $n^{(\omega/3)k+O(1)}$ may indeed be the right answer (assuming the popular conjecture that $k$-clique cannot be solved any faster). This is discussed more in the statement of \Cref{thm:main}.

\subsection{Our Techniques}

Our high-level strategy mimics that of Li~\cite{Li19}, in that we make use of the Kawarabayashi-Thorup graph sparsification technique on simple graphs, but our approach differs by exploiting matrix multiplication-based methods as well. Below, we describe these two techniques and how we apply them.

\paragraph{Kawarabayashi-Thorup Graph Sparsification}
Our first algorithmic ingredient is the (vertex) graph sparsification technique of Kawarabayashi and Thorup~\cite{kawarabayashi2018deterministic}, originally developed to solve the deterministic minimum cut problem on simple graphs. At a high level, the sparsification process contracts the input graph into one that has a much smaller number of edges and vertices such that any \emph{non-trivial} minimum cut is preserved. Here, non-trivial means that the minimum cut does not have just a singleton vertex on one side. More recently, Li~\cite{Li19} has generalized the Kawarabayashi-Thorup sparsification to also preserve non-trivial minimum $k$-cuts (those without any singleton vertices as components), which led to an $n^{(1+o(1))k}$-time minimum $k$-cut algorithm on simple graphs. The contracted graph has $\tilde O(n/\delta)$ vertices where $\delta$ is the minimum degree of the graph. If $\delta$ is large enough, say, $n^\epsilon$ for some constant $\epsilon$, then this is $\tilde O(n^{1-\epsilon})$ vertices and running the algorithm of~\cite{GHLL20} already gives $n^{(1-\epsilon)k+O(1)}$ time. At the other extreme, if there are many vertices of degree less than $n^\epsilon$, then $\lambda_k\le kn^\epsilon$ since we can take $k-1$ of these low-degree vertices as singleton components of a $k$-cut. We then employ an exact algorithm for minimum $k$-cut that runs in $\lambda_k^{O(k)}n^{O(1)}$ time (such an algorithm has been shown to exist, as we will discuss further when stating Theorem~\ref{thm:main}), which is $n^{(1-\epsilon)k+O(1)}$ time. For the middle ground where $\lambda_k > kn^\epsilon$ but there are a few vertices of degree less than $n^\epsilon$, we can modify the Kawarabayashi-Thorup sparsification in~\cite{Li19} to produce a graph of $\tilde O(n/\lambda_k)$ vertices instead, which is enough. This concludes the case when there are no singleton components of the minimum $k$-cut.

\paragraph{Matrix Multiplication}
What if the minimum $k$-cut has components that are singleton vertices? If \emph{all but one} component is a singleton, then we can use a matrix multiplication-based algorithm similar to the Ne\v{s}etril and Poljak's algorithm for $k$-clique~\cite{nevsetvril1985complexity}, which runs in $n^{(\omega/3)k+O(1)}$ time. Thus, the main difficulty is to handle minimum $k$-cuts where some components are singletons, but not $k-1$ many. The following definition will be at the core of all our discussions for the rest of this paper.
\begin{definition}[Border and Islands]
Given a $k$-cut $C$ with exactly $r$ singleton components, we denote the singleton components as $S_1=\{v_1\},\,\ldots,\,S_r=\{v_r\}$ and denote the other components $S_{r+1},\ldots,S_k$. A \textbf{border} of $C$ is a cut obtained by merging some singleton components into larger components. More preciously, a border is defined by a subset $I\subseteq [r]$ and a function $\sigma: I \to [k]\setminus[r]$. Given $I$ and $\sigma$, we let $S'_i=S_i\cup\{v_j:j\in I,\sigma(j)=i\}$, then the border $B_{I, \sigma}$ is the $(k-|I|)$-cut defined by the components $S_{r+1}', \ldots, S_{k}'$, together with the unmerged singleton components $S_j$ where $j\in [r]\setminus I$. The set of vertices $\{v_i: i\in [I]\}$, corresponding to the merged singleton components, is called the \textbf{islands}.
\end{definition}

Given this definition, our main technical contribution is as follows: we show that if the cut $C$ has exactly $r$ singleton components, then we can first apply Kawarabayashi-Thorup sparsification to compute a graph $G'$ of size $\tilde{O}(n/\lambda_k)$ that preserves some borders of $C$. We then use the algorithm of~\cite{GHLL20} on $G'$ to discover a border, which will succeed with probability roughly $1/\tilde O((n/\lambda_k)^{k-|I|})$. Finally, we run a matrix multiplication-based algorithm to locate the islands in an additional $n^{(\omega/3)|I|+O(1)}$ time. Altogether, the runtime becomes 
$\tilde O((n/\lambda_k)^{k-s}) \cdot n^{(\omega/3)s+O(1)}$, which is $n^{(1-\epsilon')k+O(1)}$ as long as $\lambda_k\ge n^\epsilon$, where $\epsilon'$ depends on $\epsilon$.

We summarize our discussions with the following theorem, which is the real result of this paper.
\begin{theorem}\label{thm:main}
Suppose there exists an algorithm that takes in a simple, unweighted graph $G$, and returns its minimum $k$-cut in time $\lambda_k^{tk}n^{O(1)}$. Let $c = \max(\frac{t}{t+1}, \frac{\omega}{3})$. Then we can compute a minimum $k$-cut of a simple, unweighted graph in $O(n^{ck + O(1)})$.
\end{theorem}
Recently, Lokshtanov, Saurabh, and Surianarayanan~\cite{lokshtanov2020parameterized} showed an algorithm for exact minimum $k$-cut that runs in time $\lambda^{O(k)}n^{O(1)}$. Combining their result with Theorem~\ref{thm:main}, we obtain a minimum $k$-cut that runs in time $O(n^{ck + O(1)})$ for some constant $c < 1$. We further note that we use their algorithm in a black-box manner, which means if one could derive an exact algorithm with a better constant $t$, then our algorithm will have an improved runtime up to $O(n^{\omega k/3 + O(1)})$.

\section{Main Algorithm}
In this section, we discuss our algorithm in detail. Given a simple, unweighted graph $G$, we first run an approximate $k$-cut algorithm to determine the magnitude of $\lambda_k$. If $\lambda_k\le O(n^{\frac{1}{t+1}})$, then we can run the exact algorithm on $G$ and output its result. Otherwise, we apply Lemma~\ref{lem:kt-main}, which is a modified version of Kawarabayashi-Thorup sparsification~\cite{kawarabayashi2018deterministic} for $k$-cuts. These modifications, discussed in Section~\ref{sec:KT}, will give us a graph $G'$ on $\tilde{O}(n/\lambda_k)$ vertices that preserves at least one border for every minimum $k$-cut of $G$. Now we fix any minimum $k$-cut of $G$, and fix its border $B_{I, \sigma}$ specified by Lemma~\ref{lem:kt-main}. For every possible value of $|I|$, we run Lemma~\ref{lem:karger} to discover $B_{I, \sigma}$ with high probability.

Once we found the border, locating the islands is simple. In Section~\ref{sec:islands}, we present a slight variant of Ne\v{s}etril and Poljak's $k$-clique algorithm~\cite{nevsetvril1985complexity} that solves the following problem in $O(n^{\frac{\omega}{3}r + O(1)})$ time. \begin{definition}[$r$-Island Problem]
Given a graph $G$, find the optimal $r+1$-cut $C$ which has exactly $r$ singleton components. 
\end{definition}
This enables us to recover the minimum $k$-cut in $G$ by guessing the number of islands in each non-singleton component specified by the border, and finding them independently. The total runtime is $O(n^{\frac{\omega}{3}|I| + O(1)})$ since the number of islands in any non-singleton component is at most the total number of islands $|I|$. This proves Theorem~\ref{thm:main}.

Our methods are summarized in the following algorithm. Note that for the initial $O(1)$-approximation step, various algorithms can be used. 

\begin{algorithm}[H]
\caption{Main Algorithm}
\label{alg:main}
\begin{algorithmic}[1]
\State Run an $2$-approximation algorithm of $k$-CUT in polynomial time~\cite{SV95}, and let its output be $\overline{\lambda_k}$. 
\If{$\overline{\lambda_k} \le 10n^{\frac{1}{t+1}}$}
    \State Run the given exact algorithm for $k$-CUT
\Else 
    \State Apply Lemma~\ref{lem:kt-main} to obtain a graph $G'$ on $O_k(n/\lambda_k)$ vertices. 
    \For{each $i=0,1,2,\ldots,r$} \Comment{Iterate over possible values of $i=|I|$ of the border $B_{I,\sigma}$}
        \State Run \Cref{lem:karger} with parameter $\beta=1-(1-2/\log n)i/k$.
        \For{each cut $C$ output by \Cref{lem:karger}} \Comment{at most $(n/\lambda_k)^{k-(1-2/\log n)i+O(1)}$ many}
            \State Guess the number of islands in each non-singleton component of $C$.
            \State Run the Island Discovery Algorithm in each non-singleton component.
        \EndFor
    \EndFor
\EndIf
\end{algorithmic}
\end{algorithm}

\subsection{Analysis}
Our analysis is divided into three parts, each corresponding to one section of the algorithm. The first part concerns the Kawarabayashi-Thorup sparsification, and the following theorem is proved in \Cref{sec:KT}.

\begin{restatable}{lemma}{KTSparsification}\label{lem:kt-main}
For any simple graph, we can compute in $k^{O(k)}n^{O(1)}$ time a partition $V_1,\ldots,V_q$ of $V$ such that $q=(k\log n)^{O(1)}n/\lambda_k$ and the following holds:
 \begin{enumerate}
 \item[$(*)$] For any minimum $k$-cut $C$ with exactly $r$ singleton components, there exists $I\subset [r]$ and a function $\sigma:I \to[k]\setminus[r]$ such that the border of $C$ defined by $I$ and $\sigma$, namely $B_{I, \sigma}$, agrees with the partition $V_1,\ldots,V_q$. In other words, all edges of $B_{I, \sigma}$ are between some pair of parts $V_i, V_j$. Moreover, we have $|B_{I, \sigma}|\le\lambda_k-(1-2/\log n)|I|\lambda_k/k$.
 \end{enumerate}
Contracting each $V_i$ into a single vertex, we obtain a graph $G'$ on $\tilde{O}(n/\lambda_k)$ vertices that preserves $B_{I,\sigma}$. 
\end{restatable}

Next, we describe and analyze the algorithm that computes the border. The following lemma is proved in \Cref{sec:border}.
\begin{restatable}{lemma}{Karger}\label{lem:karger}
Fix an integer $2\le s\le k$ and a parameter $\beta\le1$, and consider an $s$-cut $C$ of size at most $\beta\lambda_k$. There is an $n^{\beta k+O(1)}$ time algorithm that computes a list of $n^{\beta k+O(1)}$ $s$-cuts such that with high probability, $C$ is listed as one of the cuts.
\end{restatable}

Finally, we present and analyze the algorithm that extends the border by computing the missing islands in each non-singleton component. The following lemma is proved in \Cref{sec:islands}.
\begin{restatable}{lemma}{Islands}\label{lem:islands}
There is a $O_r(n^{\frac{\omega r}{3} + O(1)})$ deterministic algorithm that solves the $r$-Island problem. 
\end{restatable}
With these three lemmas in hand, we now analyze \Cref{alg:main}.

Fix a minimum $k$-cut. The initial Kawarabayashi-Thorup sparsification takes $k^{O(k)}n^{O(1)}$ time by \Cref{lem:kt-main}, and the border $B_{I,\sigma}$ is preserved by the partition and has size at most $\lambda_k-(1-2/\log n)|I|\lambda_k/k$. For the correct guess of $|I|$, \Cref{lem:karger} detects $B_{I,\sigma}$ with high probability among a collection of $(n/\lambda_k)^{k-(1-2/\log n)i+O(1)}$ many $(k-i)$-cuts. Finally, for the $(k-i)$-cut $C=B_{I,\sigma}$, the Island Discovery Algorithm extends it to a minimum $k$-cut in time $n^{(\omega/3)i+O(1)}$. The total running time is therefore
\[ k^{O(k)}n^{O(1)} + (n/\lambda_k)^{k-(1-2/\log n)i+O(1)} \cdot n^{(\omega/3)i+O(1)} \le k^{O(k)}n^{O(1)} + n^{\frac t{t-1}(k-i+2i/\log n+O(1))}\cdot n^{(\omega/3)i+O(1)} .\]
The $n^{2i/\log n}$ term is at most $O(1)^{2i}$, which is negligible. The running time is dominated by either $i=0$ or $i=k$, depending on which of $\frac t{t-1}$ and $\omega/3$ is greater. This concludes the analysis of \Cref{alg:main} and the proof of \Cref{thm:main}.

\section{Kawarabayashi-Thorup Sparsification}\label{sec:KT}

In this section, we prove the following Kawarabayashi-Thorup sparsification theorem of any simple graph. Rather than view it as a vertex sparsification process where groups of vertices are contracted, we work with the grouping of vertices itself, which is a partition of the vertex set. We use \emph{parts} to denote the vertex sets of the partition to distinguish them from the \emph{components} of a $k$-cut.

Most of the arguments in this section originate from Kawarabayashi and Thorup's original paper~\cite{kawarabayashi2018deterministic}, though we find it more convenient to follow the presentations of~\cite{GLL21} and~\cite{Li19}.

\KTSparsification*

\subsection{Regularization Step}

We first ``regularize'' the graph to obey a few natural conditions, which is done at no asymptotic cost to the number of clusters. In particular, we ensure that $m\le O(\lambda_kn)$, i.e., there are not too many edges, and $\delta\ge\lambda_k/k$, i.e., the minimum degree is comparable to the size of the $k$-cut.

\paragraph{Nagamochi-Ibaraki sparsification.}
First, we show that we can freely assume $m=O(\lambda_k n)$ through an initial graph sparsification step due to Nagamochi and Ibaraki; the specific theorem statement here is from \cite{Li19}.
\begin{theorem}[Nagamochi and Ibaraki~\cite{NI92}, Theorem~3.3 in \cite{Li19}]\label{thm:NIb}
Given a simple graph $G$ and parameter $s$, there is a polynomial-time algorithm that computes a subgraph $H$ with at most $sn$ edges such that all $k$-cuts of size at most $s$ are preserved. More formally, for all $k$-cuts $S_1,\ldots,S_k$ satisfying $|E_G[S_1,\ldots,S_k]|\le s$, we have $E_G[S_1,\ldots,S_k]=E_H[S_1,\ldots,S_k]$.
\end{theorem}

Compute a $(1+1/k)$-approximation $\tilde\lambda_k\in[\lambda_k,(1+1/k)\lambda_k]$ in time $k^{O(k)}n^{O(1)}$~\cite{lokshtanov2020parameterized}, apply \Cref{thm:NIb} with parameter $s=\tilde\lambda_k$, and replace $G$ with the returned graph $H$. This allows us to assume $m\le (1+1/k)\lambda_k n$ henceforth.

\paragraph{Lower bound the minimum degree.}
Next, we would like to ensure that the graph $G$ has minimum degree comparable to $\lambda_k$. While there exists a vertex of degree less than $\frac{\tilde\lambda_k}{(1+1/k)(k-1)}$, declare that vertex as a trivial part in the final partition, and remove it from $G$. We claim that we can remove at most $k-1$ such vertices; otherwise, the vertices together form a $k$-cut of size less than $(k-1)\cdot\frac{\tilde\lambda_k}{(1+1/k)(k-1)}=\frac{\tilde\lambda_k}{(1+1/k)}\le\lambda_k$, contradicting the value $\lambda_k$ of the minimum $k$-cut. We have thus removed at most $k-1$ vertices. The remaining task is to compute a partition of the remaining graph which has minimum degree at least $\frac{\lambda_k}{(1+1/k)(k-1)} \ge \lambda_k/k$. We then add a singleton set for each of the singleton vertices removed, which is at most $k-1$ extra parts, which is negligible since we aim for $(k\log n)^{O(1)}n/\lambda_k$ many parts in total.

\subsection{Kawarabayashi-Thorup Sparsification}
It remains to prove the following lemma, which is \Cref{lem:kt-main} with the additional assumptions $m\le2\lambda_k n$ and $\delta\ge\lambda_k/k$.
\begin{lemma}\label{lem:kt}
Suppose we are given a simple graph with $m\le2\lambda_k n$ and $\delta\ge\lambda_k/k$. Then, we can compute a partition $V_1,\ldots,V_q$ of $V$ such that $q=(k\log n)^{O(1)}n/\lambda_k$ and the following holds:
 \begin{enumerate}
 \item[$(*)$] For any minimum $k$-cut $C$ with exactly $r$ singleton components, there exists $I\subset [r]$ and a function $\sigma:I \to[k]\setminus[r]$ such that the border of $C$ defined by $I$ and $\sigma$, namely $B_{I, \sigma}$, agrees with the partition $V_1,\ldots,V_q$. In other words, all edges of $B_{I, \sigma}$ are between some pair of parts $V_i, V_j$. Moreover, we have $|B_{I, \sigma}|\le\lambda_k-(1-2/\log n)|I|\lambda_k/k$. \end{enumerate}
\end{lemma}

Our treatment follows closely from Appendix~B of~\cite{GLL21}.

\paragraph{Expander decomposition preliminaries.}
We first introduce the concept of the \emph{conductance} of a graph, as well as an expander, defined below.
\begin{definition}[Conductance]
Given a graph $G=(V,E)$, a set $S:\emptyset\subsetneq S \subsetneq V$ has \emph{conductance} \[ \frac {|\partial_GS|} {\min\{\textbf{\textup{vol}}(S),\textbf{\textup{vol}}(V\setminus S) \}} \]
in the graph $G$, where $\textbf{\textup{vol}}(S):=\sum_{v\in S}\deg(v)$.
The \emph{conductance} of the graph $G$ is the minimum conductance of a set $S\subseteq V$ in $G$.
\end{definition}
\begin{definition}
For any parameter $0<\gamma\le1$, a graph is a $\gamma$-expander if its conductance is at least $\gamma$.
\end{definition}

The following is a well-known result about decomposing a graph into expanders, for which we provide an easy proof below for convenience.
\begin{theorem}[Expander Decomposition]\label{thm:exp-decomp}
For any graph $G=(V,E)$ with $m$ edges and a parameter $\gamma<1$, there exists a partition $U_1,\ldots,U_p$ of $V$ such that:
 \begin{enumerate}
 \item For all $i\in[p]$, $G[U_i]$ is a $\gamma$-expander.
 \item $|E[U_1,\ldots,U_p]| \le O(\gamma m\log m)$.
 \end{enumerate}
\end{theorem}

\paragraph{The partitioning algorithm.}
To compute the partition $V_1,\ldots,V_q$, we execute the same algorithm from Section~B~of~\cite{GLL21}, except we add an additional step~4. Throughout the algorithm, we fix parameter $\epsilon:=1/(k\log n)$.
\begin{enumerate}
\item Compute an expander decomposition with parameter $\gamma:=1/\delta$, and let $U_1,\ldots,U_p$ be the resulting partition of $V$.
\item Initialize the set $S\gets\emptyset$, and initialize $C_i\gets U_i$ for each $i\in[p]$. While there exists some $i\in[p]$ and a vertex $v\in C_i$ satisfying $\deg_{G[C_i]}(v) \le \frac25\deg_G(v)$, i.e., vertex $v$ loses at least $\frac35$ fraction of its degree when restricted to the current $C_i$, remove $v$ from $C_i$ and add it to $S$. The set $S$ is called the set of \emph{singleton} vertices. Note that some $C_i$ can become empty after this procedure. At this point, we call each $C_i$\ a \emph{cluster} of the graph. This procedure is called the \emph{trimming} step in~\cite{kawarabayashi2018deterministic}. 
\item Initialize the set $L := \bigcup_{i\in[p]}\{v\in C_i \mid \deg_{G[C_i]}(v) \le (1-\epsilon)\deg_G(v)\}$, i.e., for each $i\in[p]$ and vertex $v\in C_i$ that loses at least $\epsilon$ fraction of its degree when restricted to $C_i$, add $v$ to $L$ (but do not remove it from $C_i$ yet). Then, add $L$ to the singletons $S$ (i.e., update $S\gets S\cup L$) and define the \emph{core} of a cluster $C_i$ as $A_i\gets C_i\setminus L$. For a given core $A_i$, let $C(A_i)$ denote the cluster whose core is $A_i$. This procedure is called the \emph{shaving} step in~\cite{kawarabayashi2018deterministic}.

\item For each core $A_i$ with at most $k$ vertices, we \emph{shatter} the core by adding $A_i$ to the singletons $S$ (i.e., update $S\gets S\cup A_i$) and updating $A_i\gets\emptyset$. This is the only additional step relative to~\cite{GLL21}.

\item Suppose there are $p'\le p$ nonempty cores $A_i$. Let us re-order the cores $A_1,\ldots,A_p$ so that $A_1,\ldots,A_{p'}$ are precisely the nonempty cores. The final partition $\mathcal P=\{V_1,V_2,\ldots\}$ of $V$ is $\bigcup_{i\in[p']}\{A_i\} \cup \bigcup_{v\in S}\{\{v\}\}$. In other words, we take each nonempty core $A_i$ as its own set in the partition, and add each vertex $v\in S$ as a singleton set. We call each nonempty core $A_i$ a \emph{core} in the partition, and each vertex $v\in S$ as a \emph{singleton} in the partition.
\end{enumerate}

The lemmas below are stated identically to those in~\cite{GLL21}, so we omit the proofs and direct interested readers to~\cite{GLL21}.



\begin{lemma}[Lemma~B.11 of~\cite{GLL21}]\label{lem:S-small}
Fix a parameter $\alpha\ge1$ that satisfies $\alpha<o(\delta/\log n)$.
For each nonempty cluster $C$ and a subset $S\subseteq V$ satisfying $|\partial _GS|\le\alpha\delta$, we have either $|C\cap S|\le3\alpha$ or $|C\setminus S|\le3\alpha$.
\end{lemma}

The lemma below from~\cite{GLL21} is true for the algorithm without step~4.
\begin{lemma}[Corollary~B.9 of~\cite{GLL21}]
Suppose we skip step~4 of the algorithm. Then, there are $O(\frac{ m\log m}{\delta^2})$ many sets in the partition $\mathcal P$.
\end{lemma}
Clearly, adding step~4 increases the number of parts by a factor of at most $k$, so the we obtain the following corollary.
\begin{corollary}
There are $O(\frac{km\log m}{\delta^2})$ many sets in the partition $\mathcal P$.
\end{corollary}
Since $m\le\lambda_k n$ and $\delta\ge\lambda_k/k$ by the assumption of \Cref{lem:kt}, this fulfills the bound $q=(k\log n)^{O(1)}n/\lambda_k$ of \Cref{lem:kt}. For the rest of this section, we prove property~$(*)$.

The following lemma is a combination of Lemma~B.12 of~\cite{GLL21} and Lemma~16 of~\cite{Li19}, and we provide a proof for completeness.
\begin{lemma}\label{lem:unique-side}
Fix a parameter $\alpha\ge1$ that satisfies $\alpha<o(\frac\delta{k\log n})$. For any nonempty core $A$ and any minimum $k$-cut of size at most $\alpha\delta$, there is exactly one component $S^*$ satisfying $|S^*\cap C(A)|>3\alpha$, and any other component $S$ that is non-singleton must be disjoint from $A$.
Moreover, each vertex $v\in A$ has at least $(1-2\epsilon)\deg(v)$ neighbors in $S^*$.
\end{lemma}
\begin{proof}
We first show that $|C(A)|>3\alpha k$. Since $C(A)$ is nonempty, each vertex $v\in C(A)$ has at least $\frac25\deg(v)\ge\frac25\delta$ neighbors in $C(A)$, so $|C(A)|\ge\frac25\delta-1>3\alpha k$ by the assumption $\alpha<o(\frac\delta{k\log n})$. 

By \Cref{lem:S-small}, each component $S$ must satisfy $|C(A)\cap S|\le3\alpha$ or $|C(A)\setminus S|\le3\alpha$, and the latter implies that $|C(A)\cap S|>|C(A)|/2$, which only one side $S$ can satisfy. Moreover, one such component $S^*$ must exist since otherwise, $|C(A)|=\sum_S|C(A)\cap S|\le3\alpha k$, a contradiction. Therefore, all but one component $S^*$ satisfy $|C(A)\cap S|\le3\alpha$. 

Next, each vertex $v\in A$ has at least $(1-\epsilon)\deg(v)$ neighbors in $C(A)$, and at most $3\alpha k$ of them can go to $C(A)\cap S$ for any component $S\ne S^*$. This leaves at least $(1-\epsilon)\deg(v)-3\alpha k$ neighbors in $S^*$, which is at least $(1-2\epsilon)\deg(v)$ since $\epsilon=1/\log n$ and $\alpha<o(\frac\delta{k\log n})$.

We now show that if $S$ is non-singleton and $|C(A)\cap S|\le3\alpha$, then $S$ is disjoint from $A$. Suppose otherwise; then, any vertex $v\in A\cap S$ has at least $(1-2\epsilon)\deg(v)$ neighbors in $S^*$ as before. If we move $v$ from $S$ to $S^*$, then the result is still a $k$-cut since $S$ is non-singleton. Moreover, the edges from $v$ to $S$ are newly cut, and the edges from $v$ to $S^*$ are saved. The former is at most $\epsilon\deg(v)+3\alpha$, and the latter at least $(1-2\epsilon)\deg(v)$. Since $\epsilon=1/\log n$ and $\alpha<o(\frac\delta{k\log n})$, the new $k$-cut is smaller than the old one, a contradiction.
\end{proof}
Finally, we prove property~$(*)$ of \Cref{lem:kt}.
\begin{lemma}\label{lem:border-size-bound}
For any minimum $k$-cut $C$ with exactly $r$ singleton components, there exists $I\subset [r]$ and a function $\sigma:I \to[k]\setminus[r]$ such that the border of $C$ defined by $I$ and $\sigma$, namely $B_{I, \sigma}$, agrees with the partition $V_1,\ldots,V_q$. In other words, all edges of $B_{I, \sigma}$ are between some pair of parts $V_i, V_j$. Moreover, we have $|B_{I, \sigma}|\le\lambda_k-(1-2/\log n)|I|\lambda_k/k$. 
\end{lemma}
\begin{proof}
Enumerate the singleton components as $S_1=\{v_1\},\ldots,S_r=\{v_r\}$. Let $T$ be the set of singleton components $S_i$ such that $S_i$ is contained in a part $V_j$ that has more vertices than just $v_i$ (i.e., $V_j\supsetneq\{v_i\}$). For every such component $S_i=\{v_i\}$, since $V_j\supsetneq\{v_i\}$, we must have $|V_j|>k$, since otherwise it would have been shattered into singletons on step~4 of the algorithm. So there must be a non-singleton component $S_{i^*}$ of the minimum $k$-cut intersecting $V_j$ (which is unique by \Cref{lem:unique-side}). This component must be the $S^*$ from \Cref{lem:unique-side}. We define $\sigma(i)=i^*$.

As we've argued in the previous paragraph, the border $S'_{r+1},\ldots,S'_k$ defined as $S'_i=S_i\cup\{v_j:j\in I,\sigma(j)=i\}$ agrees with the partition $V_1,\ldots,V_q$. It remains to show that $|E(S'_{r+1},\ldots,S'_k)|\le\lambda_k-(1-1/k)|I|\lambda_k/k$. For each component $S_i=\{v_i\}$ with $i\in I$, by \Cref{lem:unique-side}, the vertex $v_i$ has at least $(1-2\epsilon)\deg(v)$ neighbors in $S_{\sigma(i)}$, so merging $v_i$ with $S_{\sigma(i)}$ decreases the cut value by at least $(1-2\epsilon)\deg(v)$. It follows that the border has size at most $\lambda_k-(1-2\epsilon)|I|\deg(v)$, which meets the bound since $\epsilon=1/\log n$ and $\deg(v)\ge\delta\ge\lambda_k/k$ by assumption. 
\end{proof}
With \Cref{lem:border-size-bound}, this concludes the proof of \Cref{lem:kt}.

\section{Finding the Border}\label{sec:border}
In this section, we develop an algorithm to compute the border. The main lemma is the following, where $C$ represents the border we wish to find.
\Karger*

Our algorithm follows Karger's contraction algorithm, stated below, and its analysis from~\cite{GHLL20}.
\begin{algorithm}[H]
    \caption{Contraction Algorithm~\cite{GHLL20}}
    \label{euclid}
    \begin{algorithmic}[1] 
            \While{$|V| > \tau$} \label{line:1}
                        \State Choose an edge $e \in E$ at random from $G$, with probability proportional to its weight.
                \State Contract the two vertices in $e$ and remove
                self-loops.
            \EndWhile \label{line:4}
            \State Return a $k$-cut of $G$ chosen uniformly at random. 
    \end{algorithmic}
\end{algorithm}
The key lemma we use is the following from~\cite{GHLL20}.
\begin{lemma}[Lemma~17 of~\cite{GHLL20}]
\label{llem1}
Suppose that $J$ is an edge set with $\alpha = |J|/\lambda_k$ and $n \geq \tau \geq 8 \alpha k^2 + 2 k$. Then $J$ survives lines~\ref{line:1}~to~\ref{line:4} of the Contraction Algorithm with probability at least $(n/\tau)^{-\alpha k} k^{-O(\alpha k^2)}$.
\end{lemma}
The algorithm sets $\tau=8\beta k^2+2k$, and by \Cref{llem1}, any $s$-cut $C$ of size $\alpha\lambda_k$ for some $\alpha\le\beta$ survives lines~\ref{line:1}~to~\ref{line:4} of the Contraction Algorithm with probability $k^{-O(k^2)}n^{-\alpha k}\ge k^{-O(k^2)}n^{-\beta k}$. The algorithm sets $s$ for the parameter $k$, and $C$ is output with probability $1/r^\tau\ge k^{-O(k^2)}$. Overall, the probability of outputting $C$ is $k^{-O(k^2)}n^{-\beta k}$. Repeating the algorithm $k^{O(k^2)}n^{\beta k}\log n$ times, we can output a list of cuts that contains $C$ with high probability.

\section{Finding the Islands}\label{sec:islands}
In this section, we prove the following lemma. 
\Islands*

We present an algorithm for $r$-island which is a variant of Ne\v{s}etril and Poljak's $k$-clique algorithm~\cite{nevsetvril1985complexity}. Given an input graph $G$, we want to find the optimal $r$ vertices to cut off from $G$. Note that this is similar to finding the minimum $r$-clique in $G$, except that we need to take into account the edges from the $r$ islands to the remaining giant component in $G$. We first consider the case where $r$ is divisible by $3$. 

\begin{algorithm}[H]
\caption{Island Discovery Algorithm}
\label{alg:island}
\begin{algorithmic}[1]
\State We construct a weighted graph $G'$ as follows --- for every subset of vertices $S$ such that $|S| = \frac{r}{3}$, create a vertex $v_S$. Denote the total number of edges among vertices in $S$ as $w_S$, and denote the total number of edges between $S$ and $V\setminus S$ as $w_S^V$. For each pair of vertices $v_S, v_T$, let $w_{S, T}$ be the total number of edges between $S$ and $T$ if they are disjoint. Add an edge between them of weight $w_{S, T}$.
\State We want to find the minimum weight triangle in the graph $G'$. To do so, we guess the weight of a minimum weight triangle as follows: Denote the three vertices as $v_{S_1}, v_{S_2}, v_{S_3}$. Guess $w_{S_1}, w_{S_2}, w_{S_3}, w^V_{S_1}, w^V_{S_2}, w^V_{S_3}, $ and $w_{S_1, S_2}, w_{S_2, S_3}, w_{S_3, S_1}$. 
\State Denote $A$ as the binary adjacency matrix for $G'$. Let $F_i$ denotes the set of vertices $v_S$ such that $w_S = w_{S_i}, w^V_S = w^V_{S_i}$. Define $A_{1, 2}$ to be the matrix $A$ with the rows restricted to vertices in $F_1$, and columns restricted to vertices in $F_2$. Additionally, for $v_S \in F_1, v_T\in F_2$, if $w_{S, T}\ne w_{S_1, S_2}$, set $A_{1, 2}[S, T] = 0$. Define $A_{2, 3}, A_{3, 1}$ similarly.
\State Compute the matrix product $B = A_{1, 2}\times A_{2, 3}$. If there exists $v_S\in F_1$, $v_T\in F_3$ such that $B[S, T]\ne 0, A_{3, 1} = 1$, then find $v_R$ such that $A_{1,2}[S, R] = A_{2, 3}[R, T] = 1$ and return $S, R, T$. Otherwise, return Null. 
\end{algorithmic}
\end{algorithm}

\begin{claim}
Algorithm~\ref{alg:island} returns an optimal $(r+1)$-cut with $r$ islands with probability at least $\frac{1}{O(r^{15}n^3)}$. 
\end{claim}
\begin{proof}
We first note that given the nine parameters $w_{S_1}, w_{S_2}, w_{S_3}, w^V_{S_1}, w^V_{S_2}, w^V_{S_3},$ and $w_{S_1, S_2}, w_{S_2, S_3}, w_{S_3, S_1}$, the weight of the returned cut would be $w_{S_1} + w_{S_2} + w_{S_3} + w_{S_1, S_2} + w_{S_2, S_3} + w_{S_3, S_1} + (w^V_{S_1} - w_{S_1, S_2} - w_{S_3, S_1}) + (w^V_{S_2} - w_{S_1, S_2} - w_{S_2, S_3}) + (w^V_{S_3} - w_{S_2, S_3} - w_{S_3, S_1})$. In other words, the nine parameters precisely specify the weight of the returned cut. Therefore, if we guess the parameters correctly, our algorithm will return $r$-vertices that gives the minimum $r+1$ cut with $r$ islands. Note that $w_{S_1}, w_{S_2}, w_{S_3}$ and $w_{S_1, S_2}, w_{S_2, S_3}, w_{S_3, S_1}$ each have $O(r^2)$ possible values, while $w^V_{S_1}, w^V_{S_2}, w^V_{S_3}$ each have $O(rn)$ possible values. Therefore there are at most $O(r^{15}n^3)$ possible combination of values for the nine parameters, which means we guess correctly with probability as least $\frac{1}{O(r^{15}n^3)}$. The rest of the algorithm is a standard triangle detection algorithm using matrix multiplication, which has runtime $O(n^{\frac{\omega r}{3}})$. 
\end{proof}
If $r$ is not divisible by $3$, we can add up to two isolated vertices into the graph and reduce to the case where $r$ is divisible by $3$. This increase the runtime by a factor of $n^{O(1)}$. Now note that our above algorithm can be easily made deterministic by going over all $O(r^{15}n^3)$ possible combinations of the nine parameters instead of guessing them. This proves Lemma~\ref{lem:islands}.

\section*{Acknowledgements}
The authors would like to thank Anupam Gupta for many constructive discussions and comments. 

\bibliographystyle{alpha}
\bibliography{bib}
\end{document}